\documentclass[journal]{IEEEtran}

\ifCLASSINFOpdf
  \usepackage[pdftex]{graphicx}
  \usepackage{amsmath}
  \usepackage{amsfonts}
  \newtheorem{theorem}{Theorem}
  \newtheorem{proof}{Proof}[section]
  \usepackage{tikz,mathpazo}
  \usetikzlibrary{shapes.geometric, arrows}
  \usetikzlibrary{calc}
  \usepackage{multirow,bigstrut}
\else
   \usepackage[dvips]{graphicx}
\fi
\begin{document}

\title{A Robust Color Edge Detection Algorithm Based on Quaternion Hardy Filter}

\author{Wenshan Bi, Dong Cheng and Kit Ian Kou 
\thanks{This work was supported in part by the Science and Technology Development Fund, Macau SAR FDCT/085/2018/A2 and  the Guangdong Basic and Applied Basic Research Foundation (No. 2019A1515111185).}
\thanks{ W. Bi and K.I. Kou are with the Department of Mathematics, Faculty of Science and Technology, University of Macau, Macau, China (e-mail: wenshan0608@163.com; kikou@um.edu.mo).}
\thanks{D. Cheng is with the Research Center for Mathematics and Mathematics Education, Beijing Normal University at Zhuhai, 519087, China (e-mail: chengdong720@163.com).}
}

\maketitle

\begin{abstract}
This paper presents a robust filter called quaternion Hardy filter (QHF) for color image edge detection. The  QHF  can be capable of color edge feature enhancement and noise resistance. It is flexible to use  QHF by selecting suitable parameters to handle  different levels of noise.   In particular, the quaternion analytic signal, which is an effective tool in color image processing, can also be produced by quaternion Hardy filtering with specific parameters. Based on the QHF and the improved Di Zenzo gradient operator, a novel color edge detection algorithm   is proposed. Importantly, it can be efficiently implemented by using the fast discrete quaternion Fourier transform technique. The experiments demonstrate that the proposed algorithm outperforms several widely used algorithms.
\end{abstract}


\IEEEpeerreviewmaketitle

\section{Introduction}

\IEEEPARstart{E}{dge}  detection is a fundamental problem in computer vision \cite{edge}. It has a wide range of applications, including  medical imaging \cite{medical}, lane detection \cite{lane1},  face recognition \cite{face}, weed detection \cite{weed}, and deep learning, the well known method, plays an essential role in image processing and data analysis \cite{ZT1}-\cite{d4}.
Additionally, Canny, differential phase congruence (DPC) and modified differential phase congruence (MDPC) detectors et al. have drawn wide attention and achieved great success in gray-scale edge detection \cite{Canny1986}-\cite{MDPC}.
Another optional  approach of edge detection is detecting edges independently in each of the three color channels, and then obtain the final edge map by combining three single channel edge results  according to some proposed rules \cite{separate}.  However, these methods ignore the relationship between different color channels of the image. Instead of separately computing the scaled gradient for each color component, a multi-channel gradient edge detector has been widely used since it was proposed by Di Zenzo \cite{B19}. In 2012, Jin \cite{IDZ} solved the uncertainty of the Di Zenzo gradient direction and presented an improved Di Zenzo (IDZ) gradient operator, which achieves a significant improvement over DZ. However, the IDZ algorithm is not suitable for the edge detection of noisy images.

\begin{figure}[t]
 \centering
 \includegraphics[height=3.5cm,width=7cm]{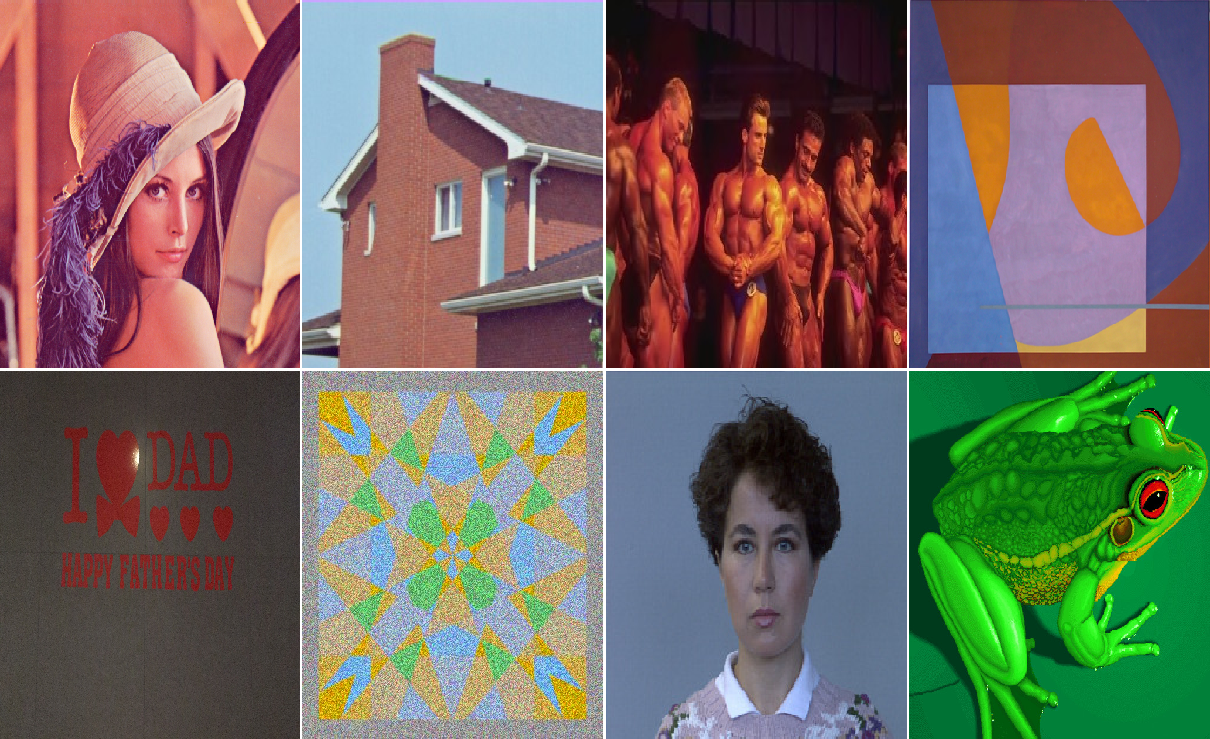}
  \caption{ The test images. From left to right and top to bottom: Lena, House, Men, T1, T2, T3, Cara and Frog.}
  \label{fig5}
\end{figure}

A growing number of research \cite{B21}-\cite{B24} indicates that quaternions are well adapted to color image processing by encoding color channels into three imaginary parts. The quaternion analytic signals are the boundary values of the functions in quaternion Hardy space \cite{B25}. Based on the quaternion analytic signal, researchers in \cite{B26} proposed some phase-based algorithms to detect the edge map of gray-scale images. It is shown that the introducing of quaternion analytic signal can reduce the influence of noise on edge detection results. It should be noted that although the tool of quaternion was applied, the algorithms (QDPC and QDPA) in \cite{B26} only considered the gray-scale images. Based on the quaternion Hardy filter and the improved Di Zenzo gradient operator, we propose a novel  edge detection algorithm,  which can be applied to color image.

The contributions of this paper are summarized as follows.
\begin{enumerate}
  \item We propose a novel filter, named quaternion Hardy filter (QHF), for color image processing. Compared with quaternion analytic signal, our method has a better performance due to the flexible parameter selection of QHF.
  \item Based on the QHF and the improved Di Zenzo gradient operator, we propose a robust color edge detection algorithm. It can enhance the color edge in a holistic manner by extracting the main features of the color image.
  \item We set up a series of experiments to verify the denoising performance of the proposed algorithm in various environments. Visual and quantitative analysis are both conducted. Three widely used edge detection algorithms, Canny, Sobel and Prewitt, and two recent edge detection algorithms, QDPC, QDPA, DPC and MDPC, are compared with the proposed algorithm.
      In terms of peak SNR (PSNR) and similarity index measure (SSIM), the proposed algorithm presents  the superiority in  edge detection.
\end{enumerate}
\begin{figure}[t]
 \centering
 \includegraphics[height=1.75cm,width=5.25cm]{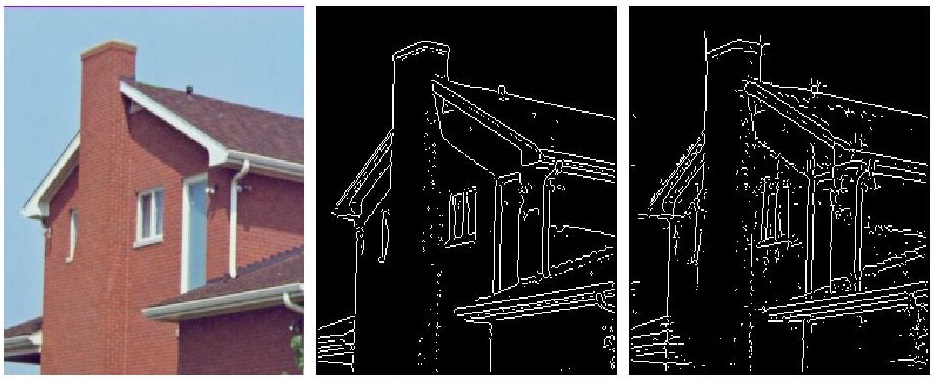}
  \caption{The noiseless House image (left). The edge maps obtained by IDZ gradient algorithm (middle) and the proposed algorithm (right).}
  \label{fig8}
\end{figure}
The rest of the paper is organized as follows. Section \ref{sec:Pre} recalls some preliminaries of the improved Di Zenzo gradient operator, quaternions, quaternion Fourier transform, quaternion Hardy space and quaternion analytic signal. Section \ref{sec:Proposed algorithm} presents the main result of the paper, it defines the novel algorithm for color-based edge detection of real-world images. Experimental results of the proposed algorithm are shown in Section \ref{sec:Exp}. Conclusions and discussions of the future work are drawn in Section \ref{sec:Con}.

\section{Preliminaries}
\label{sec:Pre}
This part recalls some preparatory knowledge of the improved Di Zenzo gradient operator \cite{IDZ}, quaternions, quaternion Fourier transform \cite{B27}, quaternion Hardy space \cite{B26} and quaternion analytic signal \cite{B28} which will be used throughout the paper.

\subsection{The improved Di Zenzo gradient operator}
\label{Pre-IDZ}
In the following, we recall the improved Di Zenzo gradient operator, namely the IDZ gradient operator, which will be combined with the quaternion Hardy filter  to establish the novel edge detection algorithm in the next section.

Let $f$  be an $M\times{N}$ color image that maps a point $(x_1,x_2)$  to a vector $(f_1(x_1,x_2),f_2(x_1,x_2),f_3(x_1,x_2))$.
Then the square of the variation of $f$ at the position $(x_1,x_2)$ with the distance $\gamma$ in the direction $\theta$ is given by
\begin{eqnarray}\label{df2}
\begin{aligned}
df^2& :=\|f(x_1+\gamma{\cos\theta,x_2+\gamma{\sin\theta})-f(x_1,x_2)}\|_2^2 \\
&\approx\sum\limits_{i=1}^{3}\left(\frac{\partial{f_i}}{\partial{x_1}}\gamma\cos\theta+
\frac{\partial{f_i}}{\partial{x_2}}\gamma\sin\theta\right)^2\\
&=\gamma^2f(\theta),
\end{aligned}
\end{eqnarray}
where
\begin{eqnarray}\label{ftheta}
 \begin{aligned}
f(\theta) :=&2\sum\limits_{i=1}^{3}\frac{\partial{f_i}}{\partial{x_1}}\frac{\partial{f_i}}{\partial{x_2}}\cos\theta\sin\theta\\
&+\sum\limits_{i=1}^{3}\left(\frac{\partial{f_i}}{\partial{x_1}}\right)^2\cos^2\theta+\sum\limits_{i=1}^{3}\left(\frac{\partial{f_i}}{\partial{x_2}}\right)^2\sin^2\theta.
\end{aligned}
\end{eqnarray}

\begin{figure}[t]
\centering
 \includegraphics[height=5.25cm,width=7cm]{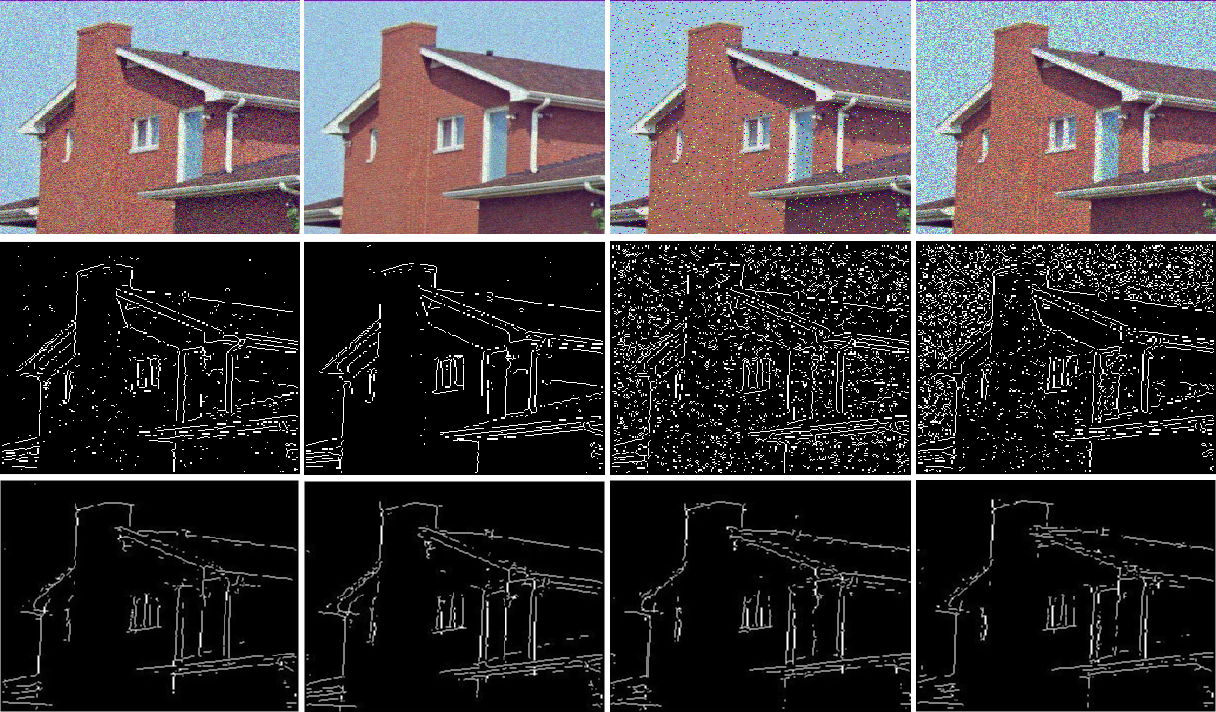}
  \caption{ The first row is the noisy House image with additive Gaussian noise, Poisson noise, Salt \& Pepper noise and Speckle noise from left to right. The second and third rows are the edge maps which are captured by IDZ  algorithm and the proposed algorithm, respectively.}
  \label{fig9}
\end{figure}
Let
\begin{eqnarray}\label{ABC}
 \left\{
 \begin{aligned}
&A:= \sum\limits_{i=1}^{3}\left(\frac{\partial{f_i}}{\partial{x_1}}\right)^2; \\
&B:= \sum\limits_{i=1}^{3}\left(\frac{\partial{f_i}}{\partial{x_2}}\right)^2; \\
&C:= \sum\limits_{i=1}^{3}\frac{\partial{f_i}}{\partial{x_1}}\frac{\partial{f_i}}{\partial{x_2}}.
\end{aligned}\right.
\end{eqnarray}
Then the gradient magnitude $f_{\max}$ of the improved Di Zenzo's gradient operator is given by
\begin{eqnarray}\label{fMAX}
\begin{aligned}
f_{\max}(\theta_{\max}):&=\max_{0\leq \theta \leq 2 \pi}{f(\theta)}\\
        =&\frac{1}{2}\bigg(A+C+\sqrt{(A-C)^2+(2B)^2}\bigg).
\end{aligned}
\end{eqnarray}
The gradient direction is defined as the value $\theta_{\max}$ that maximizes $f(\theta)$ over $0\leq \theta \leq 2 \pi$
\begin{eqnarray}\label{thetaMAX}
\theta_{\max}:=&\mbox{sgn}(B)\arcsin\bigg(\frac{f_{\max}-A}{2f_{\max}-A-C}\bigg),
\end{eqnarray}
where
$(A-C)^2+B^2\neq0$,
$\mbox{sgn}(B)=\left\{
          \begin{array}{ll}
            1, & {B\geq0;} \\
            -1, & {B<0.}
          \end{array}
        \right.$
When $(A-C)^2+B^2=0$, $\nonumber\theta_{\max}$ is undefined.

It is important to note that the IDZ edge detector is designed to process real domain signals and don't possess the capability of de-noising.

\subsection{Quaternions}
\label{sec:Pre-Q}
As a natural extension of the complex space $\mathbb{C}$, the quaternion space $\mathbb{H}$ was first proposed by Hamilton in 1843 \cite{Hamilton}. A complex number consists of two components: one real part and one imaginary part. While a quaternion $q\in\mathbb{H}$ has four components, i.e., one real part and three imaginary parts
\begin{equation}\label{Eq.H}
q=q_0+q_1\mathbf{i}+q_2\mathbf{j}+q_3\mathbf{k},
\end{equation}
where $q_{n}\in\mathbb{R}, n=0,1,2,3$, and the basis elements $\{\mathbf{i},\mathbf{j},\mathbf{k}\}$ obey the Hamilton's multiplication rules
\begin{eqnarray}\label{Eq.Hamilton's multiplication rules}
\begin{aligned}
\mathbf{i}^2&=\mathbf{j}^2=\mathbf{k}^2=\mathbf{i}\mathbf{j}\mathbf{k}=-1;\\
\mathbf{i}\mathbf{j}&=\mathbf{k},\mathbf{j}\mathbf{k}=\mathbf{i},\mathbf{k}\mathbf{i}=\mathbf{j};\\
\mathbf{j}\mathbf{i}&=-\mathbf{k},\mathbf{k}\mathbf{j}=-\mathbf{i},\mathbf{i}\mathbf{k}=-\mathbf{j}.\\
\end{aligned}
\end{eqnarray}
\begin{figure}[t]
 \centering
 \includegraphics[height=14cm,width=7cm]{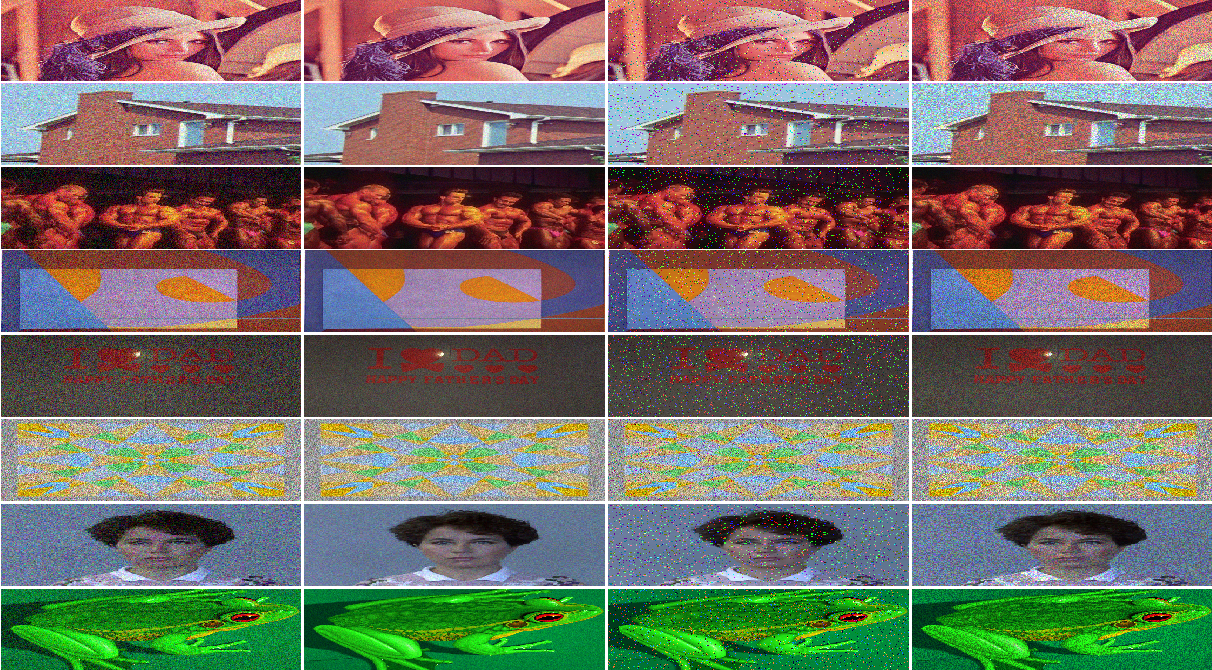}
  \caption{Noisy images. From left to right, they are obtained by adding: Gaussian noise, Poisson noise, Salt \& Pepper noise and Speckle noise to the original images (Fig. \ref{fig5}).}
  \label{fig6}
\end{figure}
Given a quaternion $q=q_0+q_1\mathbf{i}+q_2\mathbf{j}+q_3\mathbf{k}$, its quaternion conjugate is $\overline{q}:=q_0-q_1\mathbf{i}-q_2\mathbf{j}-q_3\mathbf{k}$. We write $\mathbf{Sc}(q):=\frac{1}{2}(q+\overline{q})=q_{0}$ and
$\mathbf{Vec}(q):=\frac{1}{2}(q-\overline{q})=q_1\mathbf{i}+q_2\mathbf{j}+q_3\mathbf{k}$,
which are the scalar and vector parts of $q$ , respectively. This leads to a modulus of $q\in\mathbb{H}$ defined by
\begin{equation}\label{Eq.modulus}
|q|:=\sqrt{q\overline{q}}=\sqrt{\overline{q}q}=\sqrt{q_0^{2}+q_1^{2}+q_2^{2}+q_3^{2}},
\end{equation}
where $q_{n}\in\mathbb{R}, n=0,1,2,3$.

By \eqref{Eq.H}, an $\mathbb{H}$-valued function $f:\mathbb{R}^{2}\rightarrow\mathbb{H}$ can be expressed as
\begin{eqnarray}\label{Eq.H-valued function}
\begin{aligned}
f(x_{1},x_{2})=&f_0(x_{1},x_{2})+f_1(x_{1},x_{2})\mathbf{i}+f_2(x_{1},x_{2})\mathbf{j}\\
&+f_3(x_{1},x_{2})\mathbf{k},
\end{aligned}
\end{eqnarray}
where $f_{n}:\mathbb{R}^{2}\rightarrow\mathbb{R}(n=0,1,2,3)$.

\begin{figure}[t]
 \centering
 \includegraphics[height=14.7cm,width=7cm]{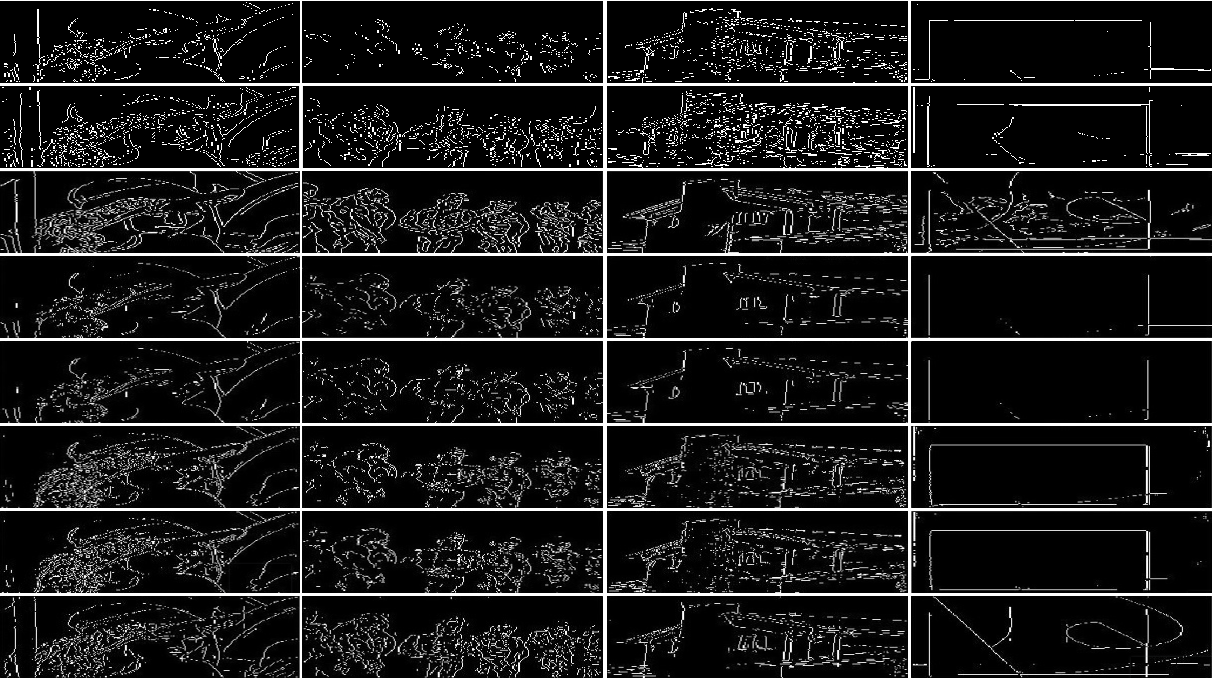}
  \caption{Edge maps of noiseless image Lena, Men , House and T1 generated by different edge detectors. From top to bottom: QDPC, QDPA, Canny, Sobel, Prewitt, DPC, MDPC and the proposed algorithm.}
  \label{6noiseless}
\end{figure}
\subsection{Quaternion Fourier transform}
\label{sec:Pre-Fourier transform}

Suppose that $f$ is an absolutely integrable complex function defined on $\mathbb{R}$, then the Fourier transform \cite{stein} of $f$ is given by
\begin{eqnarray}\label{Eq.Fourier transform}
\widehat{f}(w):=\frac{1}{\sqrt{2\pi}}\int_{\mathbb{R}}f(x)e^{-iwx}d{x},
\end{eqnarray}
where $w$ denotes the angular frequency. Moreover, if  $\widehat{f}$ is an absolutely integrable complex function defined on $\mathbb{R}$ , then $f$ can be reconstructed by the Fourier transform of $f$ and is expressed by
\begin{eqnarray}\label{Eq.inverse Fourier transform}
f(x) =\frac{1}{\sqrt{2\pi}}\int_{\mathbb{R}}\widehat{f}(w)e^{iwx}dw.
\end{eqnarray}

\begin{figure}[th]
 \centering
 \includegraphics[height=14.7cm,width=7cm]{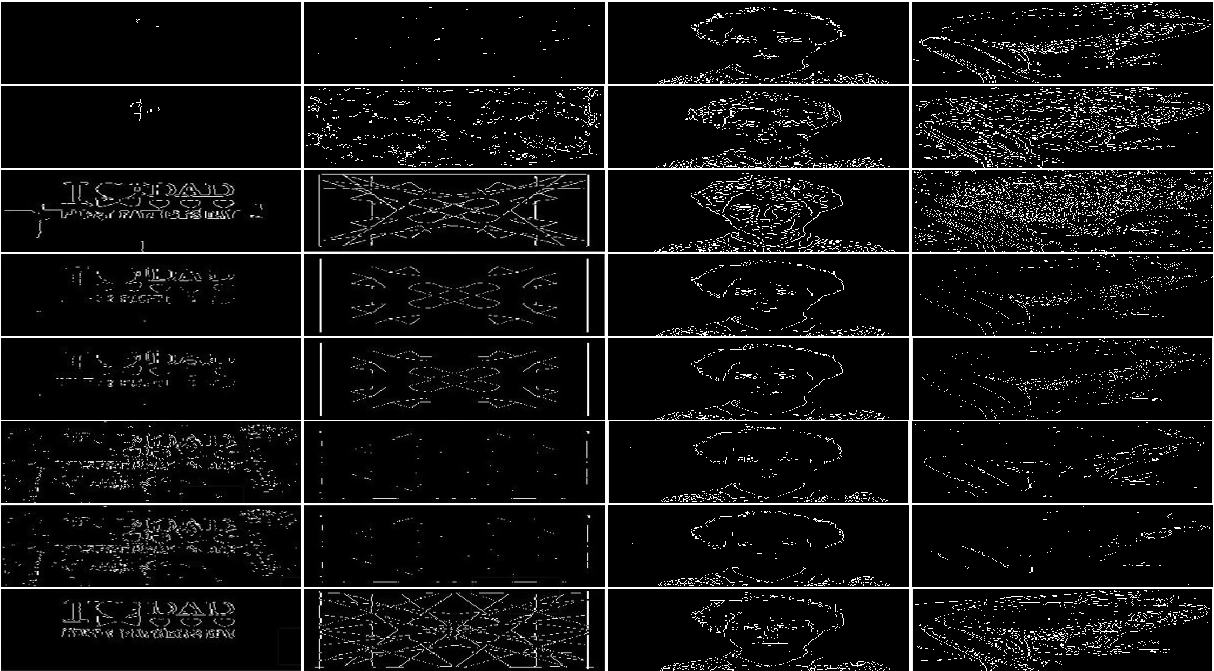}
  \caption{Edge maps of noiseless image T2, T3, Cara and Frog generated by different edge detectors. From top to bottom: QDPC, QDPA, Canny, Sobel, Prewitt, DPC, MDPC and  the proposed algorithm.}
  \label{T1-3}
\end{figure}
The quaternion Fourier transform, regarded as an extension of Fourier transform in quaternion domain, plays a vital role in grayscale image processing. The first definition of the quaternion Fourier transform was given in 1992 \cite{B31} and the first application to color images was discussed in 1996 \cite{B32}. It was recently applied to find the envelope of the image \cite{B33}. The application of quaternion Fourier transform on color images was discussed in \cite{B24}, \cite{B34}. The Plancherel and inversion theorems of quaternion Fourier transform in the square integrable signals class was established in \cite{B35}. Due to the non-commutativity of the quaternions, there are various types of quaternion Fourier transforms. In the following, we focus our attention on the two-sided quaternion Fourier transform (QFT).


\begin{figure}[ht]
 \centering
 \includegraphics[height=14.7cm,width=7cm]{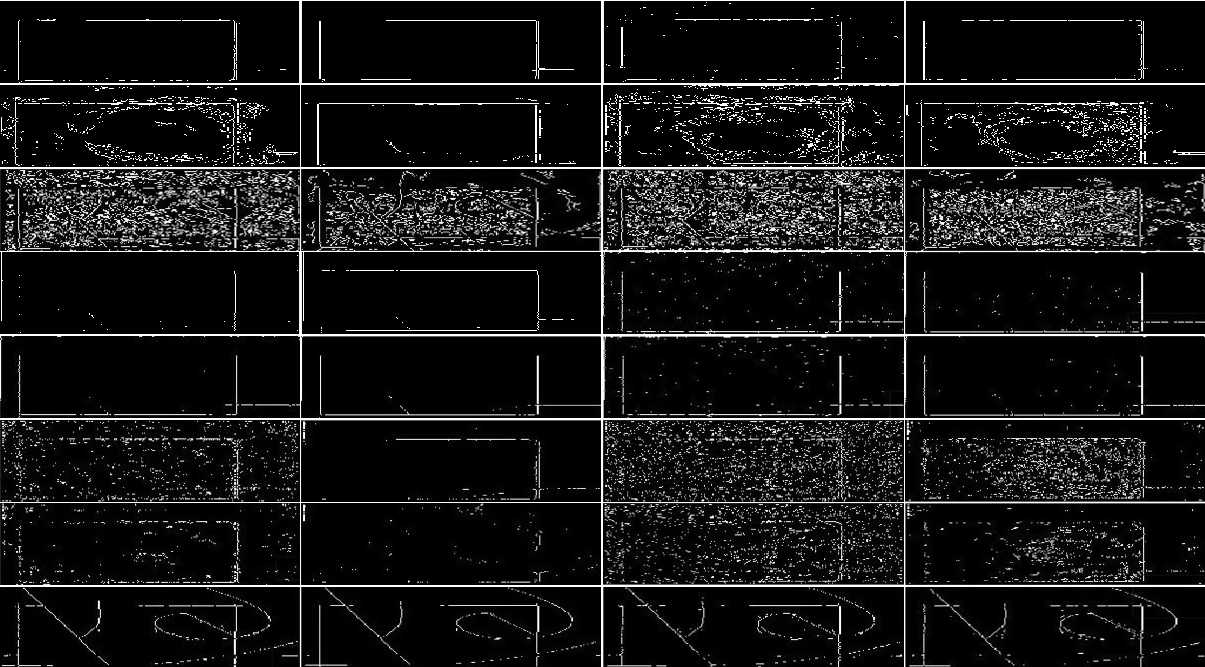}
  \caption{The edge maps of the noisy image T1 given by different algorithms. From top to bottom: QDPC, QDPA, Canny, Sobel, Prewitt, DPC, MDPC and the proposed algorithm. The first column to the last column show respectively the edge maps of T1 with Gaussian noise, Poisson noise, Salt \& Pepper noise and Speckle noise.}
  \label{nT1}
\end{figure}
Suppose that $f$ is an absolutely integrable $\mathbb{H}$-valued function defined on $\mathbb{R}^2$, then the continuous quaternion Fourier transform of $f$ is defined by
\begin{equation}\label{Eq.2-sided QFT}
(\mathcal{F}f)(w_1,w_2):=\frac{1}{2\pi}\int_{\mathbb{R}^{2}}e^{-\mathbf{i}w_1x_1}f(x_1,x_2)e^{-\mathbf{j}w_2x_2}d{x_1}d{x_2},
\end{equation}
where $w_l$ and $x_l$ denote the 2D angular frequency and 2D space ($l=1,2$), respectively.

Furthermore, if $f$ is an absolutely integrable $\mathbb{H}$-valued function defined on $\mathbb{R}^2$, then the continuous inverse quaternion Fourier transform (IQFT) of $f$ is defined by

\begin{equation}\label{Eq.inverse 2-sided QFT}
({\mathcal{F}}^{-1}f)(x_1,x_2):=\frac{1}{2\pi}\int_{\mathbb{R}^{2}}e^{\mathbf{i}w_1x_1}f(w_1,w_2)e^{\mathbf{j}w_2x_2}d{w_1}d{w_2},
\end{equation}
where $w_l$ and $x_l$ denote the 2D angular frequency and 2D space ($l=1,2$), respectively.

\begin{figure}[th]
 \centering
 \includegraphics[height=14.7cm,width=7cm]{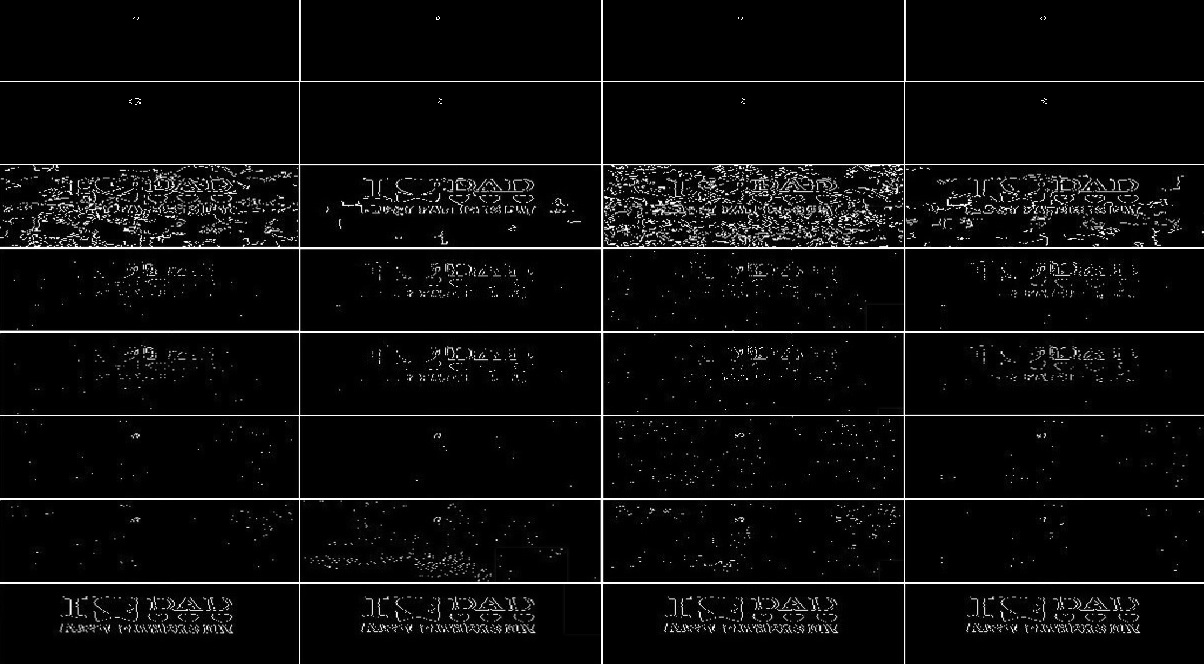}
  \caption{The edge maps of the noisy image T2 given by different algorithms. From top to bottom: QDPC, QDPA, Canny, Sobel, Prewitt, DPC, MDPC and the proposed algorithm. The first column to the last column show respectively the edge maps of T2 with Gaussian noise, Poisson noise, Salt \& Pepper noise and Speckle noise.}
  \label{nT2}
\end{figure}
The discrete quaternion Fourier transform (DQFT) and its inverse is introduced by Sangwine \cite{DQFT} in 1996. Suppose that the discrete array $f(m,n)$ is of dimension $M\times N$. The DQFT has the following form
\begin{equation}\label{DQFT}
\mathcal{F}_D[f](p,s):=\frac{1}{\sqrt{MN}}
{\sum_{m=0}^{M-1}}{\sum_{n=0}^{N-1}}
e^{-\mathbf{i}2\pi \frac{mp}{M}}f(m,n)e^{-\mathbf{j}2\pi \frac{ns}{N}}.
\end{equation}
And the inverse discrete quaternion Fourier transform (IDQFT) is
\begin{equation}\label{IDQFT}
f(m,n):=\frac{1}{\sqrt{MN}}
{\sum_{p=0}^{M-1}}{\sum_{s=0}^{N-1}}
e^{\mathbf{i}2\pi \frac{mp}{M}}\mathcal{F}_D[f](p,s)e^{\mathbf{j}2\pi \frac{ns}{N}}.
\end{equation}

\subsection{Quaternion Hardy space}
\label{sec:Quaternion Hardy space}

Let $ {\mathbb{C}}:= \{z|z=x+si, x, s \in \mathbb{R}\}$ be the complex plane and a subset of ${\mathbb{C}}$ is defined by ${\mathbb{C}}^{+}:= \{z|z=x+si,
x, s \in \mathbb{R},s>0\}$, namely upper half complex plane.
\begin{figure}[th]
 \centering
 \includegraphics[height=14.7cm,width=7cm]{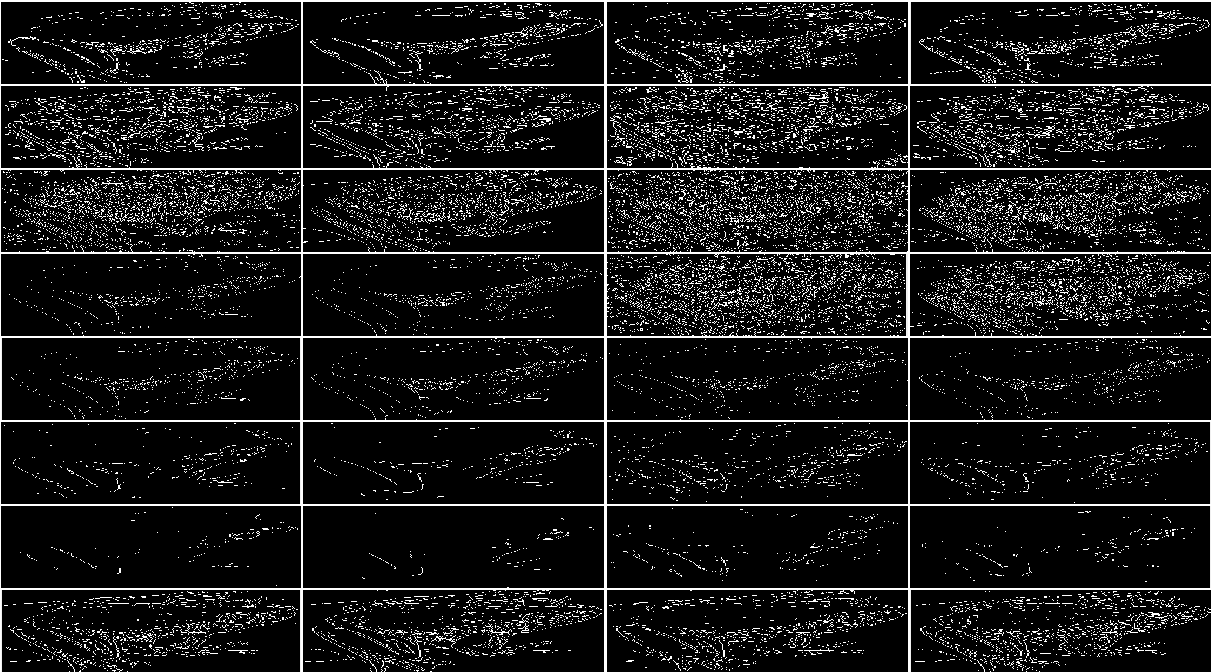}
  \caption{The edge maps of the noisy image Frog given by different algorithms. From top to bottom: QDPC, QDPA, Canny, Sobel, Prewitt, DPC, MDPC and the proposed algorithm. The first column to the last column show respectively the edge maps of Frog with Gaussian noise, Poisson noise, Salt \& Pepper noise and Speckle noise.}
  \label{nFrog}
\end{figure}
The Hardy space ${\mathbf{H}}^2(\mathbb{C}^{+})$ on the upper half complex plane consists of functions $c$ satisfying the following conditions
\begin{eqnarray}\label{Eq.Hardy space}
 \left\{
 \begin{aligned}
&\frac{\partial}{\partial \overline{z}}c(z)=0;\\
&(\sup\limits_{s>0}\int_{\mathbb{R}} |c(x+si)|^{2}dx)^{\frac{1}{2}} < \infty.
\end{aligned}\right.
\end{eqnarray}

The generalization \cite{B26} to higher dimension is given as follows. Let $ \mathbb{C}_{\mathbf{i} \mathbf{j}}:= \{( z_1, z_2)|z_1=x_1+s_1\mathbf{i}, z_2=x_2+s_2\mathbf{j},
x_l, s_l \in \mathbb{R}, l=1,2\}$ and a subset of $\mathbb{C}_{\mathbf{i} \mathbf{j}}$ is defined by $ \mathbb{C}_{\mathbf{i} \mathbf{j}}^{+}:= \{( z_1, z_2)|z_1=x_1+ s_1\mathbf{i}, z_2=x_2+ s_2\mathbf{j},x_l, s_l \in \mathbb{R}, s_l> 0, l=1,2\}.$
The quaternion Hardy space $\mathbf{Q}^2(\mathbb{C}_{\mathbf{i} \mathbf{j}}^+)$ consists of all functions satisfying the following conditions
\begin{figure}[th]
 \centering
 \includegraphics[height=5.5cm,width=7cm]{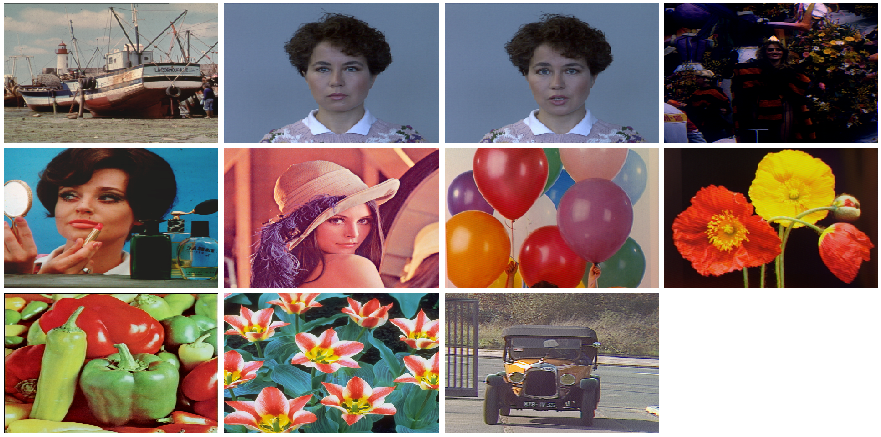}
  \caption{The test images. These images are randomly selected from the public image dataset (http://decsai.ugr.es/cvg/dbimagenes/).}
  \label{test11}
\end{figure}
\begin{eqnarray}\label{QHS conditions}
 \left\{
 \begin{aligned}
 &\frac{\partial}{\partial \overline{z_{1}}}h(z_1, z_2)=0;\\
 &h(z_1, z_2)\frac{\partial}{\partial \overline {z_{2}}}= 0;\\
 &(\sup\limits_{\substack{s_1>0 \\
 s_2>0}}\int_{\mathbb{R}^2} |h(x_1+ s_1\mathbf{i}, x_2+s_2\mathbf{j} )|^{2}dx_{1}dx_{2})^{\frac{1}{2}}  < \infty,
  \end{aligned}\right.
\end{eqnarray}
where $\frac{\partial}{\partial \overline{z_{1}}}:=\frac{\partial}{\partial {x_{1}}}+\mathbf{i}\frac{\partial}{\partial {s_{1}}}$, $\frac{\partial}{\partial \overline{z_{2}}}:=\frac{\partial}{\partial {x_{2}}}+\mathbf{j}\frac{\partial}{\partial {s_{2}}}$.

\subsection{Quaternion analytic signal}
In the following, we review the concept of analytic signal. Given a real signal $f$, combined with its own Hilbert transform, then the analytic signal of $f$ is given by
\begin{equation}\label{Eq.1D analytic signal}
f_a(x) :=f(x)+i\mathcal{H}[f](x),   x\in\mathbb{R},
\end{equation}

where $\mathcal{H}[f]$ denotes the Hilbert transform of $f$ and is defined by
\begin{equation}\label{Eq.Hilbert transform}
\mathcal{H}[f](x) :=\frac{1}{\pi} \lim\limits_{\varepsilon\to 0^+} \int_{\varepsilon \leq |x-s|} \frac{f(s)}{x-s}ds.
\end{equation}
The Fourier transform of an analytic $f_a$ defined in (\ref{Eq.1D analytic signal}) is given by
\begin{equation}\label{fa}
\widehat{f_a}(w)=\left(1+\mbox{sgn}(w)\right)\widehat{f}(w),
 \end{equation}
 where $w\in\mathbb{R}$.

A natural extension of the analytic signal from 1D to 4D space in the quaternion setting is called quaternion analytic signal. It was proposed by B\"ulow and Sommer \cite{Hypercomplex signals Sommer} using partial and total Hilbert transform associated with QFT. Given a 2D quaternion valued signal $f$, combined with its own quaternion partial and total Hilbert transform, then we get a quaternion analytic signal $f_q$ \cite{Hypercomplex signals Sommer} as follows

\begin{equation}\label{QanalyticS}
\begin{aligned}
f_q(x_1,x_2):=&f(x_1,x_2)+\mathbf{i}\mathcal{H}_{x_1}[f](x_1)\\
&+\mathcal{H}_{x_2}[f](x_2)\mathbf{j}+\mathbf{i}\mathcal{H}_{{x_1}{x_2}}[f](x_1,x_2)\mathbf{j},
\end{aligned}
\end{equation}
where
\begin{eqnarray}
\begin{aligned}
&\mathcal{H}_{x_1}[f](x_1):=
\frac{1}{\pi} \lim\int\frac{f(t_1,x_2)}{x_1-t_1}dt_1, \label{deHx1}\\
&\mathcal{H}_{x_2}[f](x_2):=
\frac{1}{\pi} \lim\int\frac{f(x_1,t_1)}{x_2-t_1}dt_1, \label{deHx2}\\
\end{aligned}
\end{eqnarray}
are the quaternion partial Hilbert transform of $f$ along the $x_1$-axis, $x_2$-axis, respectively. While
\begin{equation}\label{deHx1,x2}
\begin{aligned}
\mathcal{H}_{{x_1}{x_2}}[f](x_1,x_2):=
\frac{1}{\pi} \lim\int\frac{f(t_1,t_2)}{(x_1-t_1)(x_2-t_2)}dt_1dt_2,
\end{aligned}
\end{equation}
is the quaternion total Hilbert transform along the $x_1$ and $x_2$ axes.
By  direct computation, the quaternion Fourier transform of quaternion analytic signal   is given by
\begin{equation}\label{fq}
\begin{aligned}
(\mathcal{F}{f_q})(w_1,w_2)=&[1+\mbox{sgn}(w_1)][1+\mbox{sgn}(w_2)]\\
&(\mathcal{F}{f})(w_1,w_2).
\end{aligned}
\end{equation}

\section{Proposed algorithm}
\label{sec:Proposed algorithm}
In this section, we introduce our new color edge detection algorithm. To begin with, the definition of quaternion Hardy filter is presented.
\subsection{Quaternion Hardy filter}

The quaternion analytic signal $f_q$ can be regarded as the output signal of a filter with input $f$. The system function of this filter is
\begin{equation}\label{filter1}
H_1(w_1,w_2):=[1+\mbox{sgn}(w_1)][1+\mbox{sgn}(w_2)].
\end{equation}
In this paper, we use a novel filter, named {\it quaternion Hardy filter (QHF)}, to construct a high-dimensional analytic signal.  The system function of QHF is defined by
\begin{equation}\label{filter2}
\begin{aligned}
H(w_1,w_2,s_1,s_2):=&[1+\mbox{sgn}(w_1)][1+\mbox{sgn}(w_2)]\\
&e^{-\mid{w_1}\mid{s_1}}e^{-\mid{w_2}\mid{s_2}},
\end{aligned}
\end{equation}
where $s_1\geq0, s_2\geq0$ are parameters of the system function. The factors $(1+\mbox{sgn}(w_1))(1+\mbox{sgn}(w_2))$ and $e^{-\mid{w_1}\mid{s_1}}e^{-\mid{w_2}\mid{s_2}}$ play different roles in quaternion Hardy filter. The former performs Hilbert transform on the input signal, while the later plays a role of suppressing the high-frequency. On the one hand, the Hilbert transform operation can  selectively emphasize the edge feature of
an input object. On the other hand, the low-pass filtering can improve the ability of noise immunity for the QHF. It can be seen that as increase with $s_1, s_2$, the effect of inhibiting for the high frequency becomes more obvious. In particular, if $s_1=s_2=0$, then $e^{-\mid{w_1}\mid{s_1}}e^{-\mid{w_2}\mid{s_2}}=1$, it follows that
\begin{equation}\label{filter2=0}
H(w_1,w_2,0,0)=H_1(w_1,w_2),
\end{equation}
which means that there is no effect in high frequency inhibiting.

Parameters $s_1$ and $s_2$ play the role of low-pass filtering in vertical and horizontal directions, respectively. When the signal frequencies in the two directions are similar, then $s_1$ and $s_2$ can be set to the same value. If the signal frequencies in these two directions are different, then $s_1$ and $s_2$ should be different. For example, if the horizontal noise in the image is large, the value of $s_2$ should be set larger to enhance the anti-noise ability in that direction.
This means that the QHF is very general and flexible, and it can solve many problems that can't be solved well by quaternion analytic signal.
\label{sec:Exp-Quant}
\begin{table*}[t]
\caption{The PSNR comparison values of  Lena, Men,  House and T1.
 Types of noise: I- Gaussian noise, II- Poisson noise, III- Salt \& Pepper noise and IV- Speckle noise.}
\label{tabP1}
\centering
\begin{tabular}{ccccccccccc}
\hline
& &QDPC &QDPA     &Canny   &Sobel&Prewitt&MDPC  &DPC &IDZ &Ours\\
\hline
\multirow{4}*{LENA}
&\uppercase\expandafter{\romannumeral1}
  &58.1631 &56.8276 	&56.6985	&64.7819	&64.6047	        &57.6375	&57.6707 &53.9715 &\textbf{64.8622}\\
&\uppercase\expandafter{\romannumeral2}
  &58.1378 &56.9719 	&57.5531	&68.2679	&\textbf{68.1332}	&59.9609	&60.0363 &56.0557 &65.8570\\
&\uppercase\expandafter{\romannumeral3}
  &58.1625 &56.8949 	&55.8291	&62.4449	&62.6426	        &57.7647	&57.6952 &54.2721 &\textbf{63.0016}\\
&\uppercase\expandafter{\romannumeral4}
  &58.1266 &58.8760 	&56.4122	&64.4007	&64.3951	        &58.8433	&58.7682 &53.8232 &\textbf{64.6239}\\
\hline					
\multirow{4}*{MEN} &\uppercase\expandafter{\romannumeral1}
&58.8880 &57.3084 	&59.1716	&62.7258	        &62.7488	&58.7690	&58.8037 &53.7155 &\textbf{62.9965}\\
&\uppercase\expandafter{\romannumeral2}
&58.8770 &57.4573 	&63.4944	&\textbf{66.3657}	&66.2739	&60.9465	&61.0934 &59.2063 &65.2541\\
&\uppercase\expandafter{\romannumeral3}
&58.8746 &57.1006 	&57.3177	&58.9356	        &59.2481	&58.2379	&58.2515 &54.4117 &\textbf{62.0697}\\
&\uppercase\expandafter{\romannumeral4}
&58.8770 &57.4397 	&60.6136	&62.9773	        &62.9611	&60.3605	&60.4286 &56.3922 &\textbf{63.0900}\\
\hline							
\multirow{4}*{HOUSE}&\uppercase\expandafter{\romannumeral1}
&59.1106 &56.6819 	&56.3984	&64.9506	&\textbf{65.1998}	&57.7217	&57.7132 &54.0252 &61.4998\\
&\uppercase\expandafter{\romannumeral2}
&59.1389 &56.9562 	&63.4558	&67.6564	&\textbf{67.9324}	&60.0384	&60.1424 &56.0046 &63.4681\\
&\uppercase\expandafter{\romannumeral3}
&59.0990 &56.5921 	&56.7973	&62.9551	&\textbf{63.1633}	&58.0498	&57.9419 &54.3644 &61.0881\\
&\uppercase\expandafter{\romannumeral4}
&59.0520 &56.2835 	&56.0536	&63.8801	&\textbf{64.0296}	&58.2135	&58.0122 &53.7992 &61.0322\\
\hline
\multirow{4}*{Image T1}
&\uppercase\expandafter{\romannumeral1}
&64.9126 &62.7373 	&54.6441	&71.3263	&72.8126	&64.3338	&61.0803 &53.9459 &\textbf{69.7830}\\
&\uppercase\expandafter{\romannumeral2}
&65.2066 &62.2319 	&57.7629	&78.5871	&77.5450	&68.2203	&69.1860 &57.4738 &\textbf{74.1626}\\
&\uppercase\expandafter{\romannumeral3}
&65.2066 &62.7854 	&54.8312	&67.8883	&68.4996	&59.1397	&57.6659 &54.3844 &\textbf{68.8984}\\
&\uppercase\expandafter{\romannumeral4}
&65.0181 &62.4863 	&55.7384	&71.3680	&72.0796	&62.1912	&59.4687 &53.8051 &\textbf{70.2784}\\
\hline					
\end{tabular}
\end{table*}

\begin{table*}[th]
\caption{The PSNR comparison values of T2, T3, Cara and Frog.
 Types of noise: I- Gaussian noise, II- Poisson noise, III- Salt \& Pepper noise and IV- Speckle noise.}
\label{tabP2}
\centering
\begin{tabular}{cccccccccc}
\hline
&   &QDPC &QDPA  &Canny   &Sobel &Prewitt  &DPC &MDPC &Ours\\
\hline	
\multirow{4}*{Image T2}
&\uppercase\expandafter{\romannumeral1}	&82.3162 &76.5637 &57.9033	&67.5161  &62.0191	&61.7502 &61.6269	 &\textbf{67.7956}\\
&\uppercase\expandafter{\romannumeral2}	&82.1459 &79.1356 &66.8262	&70.8302  &63.8282  &61.8739 &60.8068	 &\textbf{70.8856}\\
&\uppercase\expandafter{\romannumeral3}	&77.6624 &70.3738 &55.8679	&66.7100  &60.4750  &61.2672 &61.0374	 &\textbf{66.8058}\\
&\uppercase\expandafter{\romannumeral4}	&82.3162 &75.1562 &61.7685	&68.6092  &61.7823	&61.8209 &61.7823	 &\textbf{68.7155}\\
\hline								
\multirow{4}*{Image T3}
&\uppercase\expandafter{\romannumeral1}	&75.3200 &59.8140 &53.3384	&61.5886  &61.6987	&62.8951 &62.9903	 &\textbf{62.9955}\\
&\uppercase\expandafter{\romannumeral2}	&75.3265 &60.648 &54.0415	&62.2696  &62.2132	&66.3745 &66.1757	 &\textbf{64.3923}\\
&\uppercase\expandafter{\romannumeral3}	&74.4772 &59.1910 &53.7885	&58.5834  &58.4737	&63.8012 &63.8826	 &\textbf{61.8965}\\
&\uppercase\expandafter{\romannumeral4}	&74.7422 &58.5271 &53.0448	&61.7655  &61.8349	&62.0452 &61.8975	 &\textbf{63.0900}\\
\hline
\multirow{4}*{Image Cara}
&\uppercase\expandafter{\romannumeral1}	&62.4899 &59.2739 &55.6093	&64.7180  &64.7514	&64.1129 &62.9773	 &\textbf{65.9913}\\
&\uppercase\expandafter{\romannumeral2}	&63.0755 &61.1621 &60.2663	&66.1169  &66.3480	&64.3532 &63.4694	 &\textbf{66.7686}\\
&\uppercase\expandafter{\romannumeral3}	&60.7350 &57.7018 &55.2998	&55.0236  &63.4468	&63.3421 &61.7183	 &\textbf{63.4500}\\
&\uppercase\expandafter{\romannumeral4}	&62.4128 &59.9237 &55.4111	&55.4800  &64.7758	&64.0322 &62.9250	 &\textbf{64.8002}\\
\hline
\multirow{4}*{Image Frog}
&\uppercase\expandafter{\romannumeral1}	&59.5452 &57.6452 &57.9129 &63.3546	&64.3588 &60.8068 &59.8769	  &\textbf{64.3217}\\
&\uppercase\expandafter{\romannumeral2}	&60.0641 &59.0153 &61.1945 &65.5119	&65.5584 &60.8961 &60.2096	  &\textbf{65.8411}\\
&\uppercase\expandafter{\romannumeral3}	&58.6196 &56.5482 &56.3742 &56.3881	&63.3670 &60.3883 &59.1801	  &\textbf{63.8809}\\
&\uppercase\expandafter{\romannumeral4}	&59.2447 &57.5758 &57.8310 &57.2878	&64.0090 &60.6560 &59.6291	  &\textbf{63.9898}\\
\hline
\end{tabular}
\end{table*}

For any fixed $s_1, s_2 \geq 0$, denote by $f_H(x_1,x_2,s_1,s_2)$ the output signal of the QHF with input $f(x_1,x_2)$. By the definition, we have
\begin{equation}\label{QFT f_H}
\begin{aligned}
(\mathcal{F}{f_H})(w_1,w_2,s_1,s_2)= [1+\mbox{sgn}(w_1)][1+\mbox{sgn}(w_2)]\\
e^{-\mid{w_1}\mid{s_1}}e^{-\mid{t_2}\mid{s_2}}(\mathcal{F}{f})(w_1,w_2).
\end{aligned}
\end{equation}
Here, the QFT acts on the variable $x_1,x_2$. We will show that as a function of $z_1=x_1+\mathbf{i}s_1$ and $ z_2=x_2+\mathbf{j}s_2$, $f_H$ belongs to the quaternion Hardy space $\mathbf{Q}^2(\mathbb{C}_{\mathbf{i} \mathbf{j}}^+)$.
\begin{theorem}
Let $f\in L^2(\mathbb{R}^2, \mathbb{H})$ and  $f_H$ be given above. Then   $f_H\in\mathbf{Q}^2(\mathbb{C}_{\mathbf{i} \mathbf{j}}^+)$ .
\end{theorem}

\begin{proof}
Using inverse quaternion Fourier transform defined by Eq. (\ref{Eq.inverse 2-sided QFT}), we have that
\begin{equation}\label{proof}
\begin{aligned}
f_H(x_1,x_2,s_1,s_2)=&\frac{1}{2\pi}\int_{\mathbb{R}^{2}}e^{\mathbf{i}w_1x_1}[1+\mbox{sgn}(w_1)]\\
&[1+\mbox{sgn}(w_2)]e^{-\mid{w_1}\mid{s_1}}e^{-\mid{w_2}\mid{s_2}}\\
&(\mathcal{F}{f})(w_1,w_2)e^{\mathbf{j}w_2x_2}d{w_1}d{w_2}.
\end{aligned}
\end{equation}
Taking the derivative of $f_H$ with respect to $\overline{z_{1}}$, we get
\begin{equation}\label{proof1}
\begin{aligned}
&\frac{\partial}{\partial \overline{z_{1}}}f_H(x_1,x_2,s_1,s_2)\\
=&[\frac{\partial}{\partial {x_{1}}}+\mathbf{i}\frac{\partial}{\partial {s_{1}}}]f_H(x_1,x_2,s_1,s_2)\\
=&\frac{1}{2\pi}\int_{\mathbb{R}^{2}}\mathbf{i}(w_1-|w_1|)
e^{\mathbf{i}w_1x_1}[1+\mbox{sgn}(w_1)][1+\mbox{sgn}(w_2)]\\
&e^{-\mid{w_1}\mid{s_1}}e^{-\mid{w_2}\mid{s_2}}(\mathcal{F}{f})(w_1,w_2)e^{\mathbf{j}w_2x_2}d{w_1}d{w_2}\\
=&0.
\end{aligned}
\end{equation}
The last equality holds since the integrand vanishes identically. Similarly,
\begin{equation}\label{proof2}
\begin{aligned}
&f_H(x_1,x_2,s_1,s_2)\frac{\partial}{\partial \overline{z_{2}}}\\
=&\frac{\partial}{\partial {x_{2}}}f_H(x_1,x_2,s_1,s_2)+\frac{\partial}{\partial {s_{2}}}f_H(x_1,x_2,s_1,s_2)\mathbf{j}\\
=&\frac{1}{2\pi}\int_{\mathbb{R}^{2}}e^{\mathbf{i}w_1x_1}[1+\mbox{sgn}(w_1)][1+\mbox{sgn}(w_2)]\\
&e^{-\mid{w_1}\mid{s_1}}e^{-\mid{w_2}\mid{s_2}}(\mathcal{F}{f})(w_1,w_2)\mathbf{j}(w_2-|w_2|)e^{\mathbf{j}w_2x_2}d{w_1}d{w_2}\\
=&0.
\end{aligned}
\end{equation}
For any fixed $s_1>0, s_2>0$, from (\ref{QFT f_H}) we can obtain that
\begin{equation}\label{proof31}
(\mathcal{F}{f_H})(w_1,w_2,s_1,s_2)\leq 4(\mathcal{F}{f})(w_1,w_2).
\end{equation}
According to the QFT Parseval's identity \cite{Cheng QFT}, we obtain that
\begin{equation}\label{proof32}
\begin{aligned}
&\int_{\mathbb{R}^{2}}|(\mathcal{F}{f_H})(w_1,w_2,s_1,s_2)|^2dw_1dw_2\\
=&\int_{\mathbb{R}^{2}} |f_H(x_1,x_2,s_1,s_2)|^2dx_1dx_2,
\end{aligned}
\end{equation}
\begin{equation}\label{proof33}
\begin{aligned}
&\int_{\mathbb{R}^{2}}|(\mathcal{F}{f})(w_1,w_2)|^2dw_1dw_2\\
=&\int_{\mathbb{R}^{2}} |f(x_1,x_2)|^2dx_1dx_2.
\end{aligned}
\end{equation}
 Using (\ref{proof31}), (\ref{proof32}) and (\ref{proof33}), a direct computation shows that
\begin{equation}\label{proof34}
\begin{aligned}
 &\sup\limits_{\substack{s_1>0 \\
 s_2>0}}\int_{\mathbb{R}^2} |f_H(x_1+ s_1\mathbf{i}, x_2+s_2\mathbf{j} )|^{2}dx_{1}dx_{2} \\
 &\leq   \sup\limits_{\substack{s_1>0 \\  s_2>0}} 16\int_{\mathbb{R}^{2}}|(\mathcal{F}{f})(w_1,w_2)|^2dw_1dw_2
 < \infty.
\end{aligned}
\end{equation}
The proof is complete.
\end{proof}

\subsection{Color edge detection algorithm}
\label{sec:algorithm}
In this section, the edge detector based on QHF  are described. The IDZ approach for edge detection based on the QHF consists in using (\ref{DQFT}), (\ref{IDQFT}) and (\ref{filter2}) to obtain the high-dimensional  analytic signal, and then use it as inputs for an appropriately defined  robust edge detection algorithm.
Let us now give the details of the quaternion Hardy filter based algorithm. They are divided by the following steps.

\begin{itemize}
  \item []
  {\bf{Step} 1}. Given an input digital color image of size $M\times N$, associate it with a $\mathbb{H}$-valued signal
  \begin{eqnarray}\label{f}
  f:=f_1\mathbf{i}+f_2\mathbf{j}+f_3\mathbf{k},
  \end{eqnarray}
  where $f_1,f_2$ and $f_3$ represent respectively the red, green and blue components of the given color image.
  \item []
  {\bf{Step} 2}. Compute the DQFT of the $f$ using equation (\ref{DQFT}). The result will be $\mathcal{F}_D[f]$.
  \item []
  {\bf{Step} 3}. For fixed $s_1>0, s_2>0$ (the values of parameters $s_1$ and $s_2$ of the original image ranged from 1.0 to 2.0, and those of the noisy image ranged from 1.0 to 8.0, for details please refer Table \ref{tab1-parameter-LENA}-\ref{tab1-parameter-FROG}), multiplying $\mathcal{F}_D[f]$ by the system function (\ref{filter2}) of the QHF. Then we obtain the DQFT of $f_H$ which has the following form
  \begin{equation}
  \begin{aligned}
  (\mathcal{F}_D{f_H})(w_1,w_2,s_1,s_2)= [1+\mbox{sgn}(w_1)]\\
  [1+\mbox{sgn}(w_2)]
  e^{-\mid{w_1}\mid{s_1}}e^{-\mid{t_2}\mid{s_2}}
  (\mathcal{F}_D{f})(w_1,w_2).
  \end{aligned}
  \end{equation}

  This is the most significant step in our algorithm, because it allows the advantages of QHF to be presented.

  \item []
  {\bf{Step} 4}. Compute the inverse DQFT for  $\mathcal{F}_D[f_H]$ by applying equation (\ref{IDQFT}), we obtain $f_H$.
  \item []
  {\bf{Step} 5}.  Extract the vector part of $f_H$, we obtain
  \begin{eqnarray}
  \mathbf{Vec}(f_H)=h_1\mathbf{i}+h_2\mathbf{j}+h_3\mathbf{k},
  \end{eqnarray}
  where $h_k$, $k=1,2,3$ are real-valued functions.
  In the following, we will operate IDZ algorithm based on $\mathbf{Vec}(f_H)$ instead of $f$.
  \item []
  {\bf{Step} 6}. Perform the IDZ gradient operator to $\mathbf{Vec}(f_H)$. Applying equation (\ref{ABC}), we obtain
  \begin{eqnarray}\label{ABC-h}
  \left\{
  \begin{aligned}
  &A= \sum\limits_{k=1}^{3}(\frac{\partial{h_k}}{\partial{x_1}})^2; \\
  &B= \sum\limits_{k=1}^{3}(\frac{\partial{h_k}}{\partial{x_2}})^2; \\
  &C= \sum\limits_{k=1}^{3}\frac{\partial{h_k}}{\partial{x_1}}\frac{\partial{h_k}}{\partial{x_2}},
  \end{aligned}\right.
  \end{eqnarray}

  then we substitute them into equation (\ref{fMAX}), obtain
  \begin{equation}\label{fMAX1}
    \mathbf{Vec}(f_H)_{\max}=\frac{1}{2}\bigg(A+C+\sqrt{(A-C)^2+(2B)^2}\bigg).
  \end{equation}

  \item []
  {\bf{Step} 7}. Finally, we obtain the processed result, edge map, by applying the nonmaxmum suppress.
  \label{algorithm}
\end{itemize}

\section{Experimental results}
\label{sec:Exp}
In this section, we shall demonstrate the effectiveness of the proposed algorithm for color image edge detection.

Here both visual and quantitative analysis for edge detection are considered in our experiments. All experiments are programmed in Matlab R2016b.
To validate the effectiveness of the proposed method, we have carried out verification on many images, eight of which are shown in Fig. \ref{fig5}. The images are from the public image dataset (http://decsai.ugr.es/cvg/dbimagenes/), which has been used by previous researchers.  It consists of 805 test images with 3 different size scales. There are 3 and 6 classes in
color images and gray images, respectively.
Here, the Gaussian filter \cite{gaussian1, gaussian2} is applied to these algorithms (Canny, Sobel, Prewitt, DPC and MDPC), since they doesn't have the ability of resisting noise.
Digital images distorted with different types of noise such as I- Gaussian noise \cite{gaussian}, II- Poisson noise, III- Salt \& Pepper noise and IV- Speckle noise. The ideal noiseless (Fig. \ref{fig5}) and noisy images (Fig. \ref{fig6}) are both taken into account.
\begin{table*}[t]
\caption{The SSIM comparison values of Lena, Men,  House and T1.
 Types of noise: I- Gaussian noise, II- Poisson noise, III- Salt \& Pepper noise and IV- Speckle noise.}
\label{tabS1}
\centering
\begin{tabular}{cccccccccc}
\hline
 &  &QDPC &QDPA  &Canny   &Sobel  &Prewitt    &DPC &MDPC &Ours \\
\hline
\multirow{4}*{LENA}
&\uppercase\expandafter{\romannumeral1}	&0.4346 &0.3473 &0.6568	&{0.8041}	&{0.8011}	&0.2876 &0.2986	 &\textbf{0.8058}\\
&\uppercase\expandafter{\romannumeral2}	&0.4365 &0.3641 &0.8593	&{0.9165}	&{0.9124}	&0.6198 &0.6200	 &\textbf{0.9275}\\
&\uppercase\expandafter{\romannumeral3}	&0.4331 &0.3632 &0.5299	&{0.5406}	&{0.5611}	&0.1711 &0.1710	 &\textbf{0.7155}\\
&\uppercase\expandafter{\romannumeral4}	&0.4301 &0.3579 &0.5658	&{0.7604}	&{0.7613}	&0.4953 &0.4959	 &\textbf{0.7872}\\
\hline
\multirow{4}*{MEN}
&\uppercase\expandafter{\romannumeral1}	&0.4736 &0.3644 &{0.6850}	&0.6673	&{0.6722}	&0.4342 &0.4343	 &\textbf{0.6878}\\
&\uppercase\expandafter{\romannumeral2}	&0.4723 &0.3741 &0.8739	&{0.8575}	&{0.8572}	&0.6523 &0.6522	 &\textbf{0.8843}\\
&\uppercase\expandafter{\romannumeral3}	&{0.4629} &0.3379 &{0.3973}	&0.1668	&0.1726	&0.1375 &0.1376	 &\textbf{0.4669}\\
&\uppercase\expandafter{\romannumeral4}	&0.4730 &0.3744 &\textbf{0.7670}	&0.7200	&{0.7281}	&0.5633 &0.5677	 &{0.7463}\\
\hline
\multirow{4}*{HOUSE}
&\uppercase\expandafter{\romannumeral1}	&0.5728 &0.6699 &0.4250	&{0.8475}	&\textbf{0.8522} &0.2289 &0.2292 	&{0.8518}\\
&\uppercase\expandafter{\romannumeral2}	&0.6015 &0.7218 &0.8503	&{0.9183}	&{0.9134}	         &0.5101 &0.5106   &\textbf{0.9192}\\
&\uppercase\expandafter{\romannumeral3}	&0.3814 &{0.5899} &0.3759	&0.5767	&{0.5921}	         &0.1419 &0.1423 	&\textbf{0.6723}\\
&\uppercase\expandafter{\romannumeral4}	&0.4343 &0.6188 &0.3387	&{0.7502}	&\textbf{0.7756} &0.2691 &0.2672 	&{0.6707}\\
\hline
\multirow{4}*{Image T1}
&\uppercase\expandafter{\romannumeral1}  &0.8417 &0.7673 &0.1449	&\textbf{0.9562}       &{0.9309} &0.3935 &0.7128  &{0.9196}\\
&\uppercase\expandafter{\romannumeral2}  &0.8441 &0.7674 &0.4470	&{0.9270}	   &{0.9456} &0.9191 &0.8435	 &\textbf{0.9652}\\
&\uppercase\expandafter{\romannumeral3}  &{0.8440} &0.7617 &0.1227	&0.7497	        &{0.8045}	&0.0637 &0.1679	 &\textbf{0.8833}\\
&\uppercase\expandafter{\romannumeral4}  &0.8454 &0.7578 &0.3621	&{0.8970}	        &{0.9121}	&0.4338 &0.6074	 &\textbf{0.9139}\\
\hline
\end{tabular}
\end{table*}
\begin{table*}[th]
\centering
\caption{The SSIM comparison values for T2, T3, Cara and Frog.
 Types of noise: I- Gaussian noise, II- Poisson noise, III- Salt \& Pepper noise and IV- Speckle noise.}
\label{tabS2}
\begin{tabular}{cccccccccc}
\hline
&  &QDPC &QDPA   &Canny &Sobel &Prewitt   &DPC &MDPC   &Ours\\
\hline
\multirow{4}*{Image  T2}
&\uppercase\expandafter{\romannumeral1} &0.9957 &0.9815	&0.4350	&0.8217	&0.8314 &0.3271 &0.3204	   &\textbf{0.9248}\\
&\uppercase\expandafter{\romannumeral2} &0.9954 &0.9926	&0.9034	&0.9202	&0.9311 &0.5865 &0.5635	   &\textbf{0.9664}\\
&\uppercase\expandafter{\romannumeral3} &0.9674 &0.9428	&0.1942	&0.7720	&0.7880	&0.3206 &0.3153	   &\textbf{0.9098}\\
&\uppercase\expandafter{\romannumeral4} &0.9957 &0.9819	&0.7197	&0.8914	&0.8901	&0.3780 &0.3584	   &\textbf{0.9483}\\
\hline
\multirow{4}*{Image  T3}
&\uppercase\expandafter{\romannumeral1} &0.9456 &0.4588	&0.0827	&0.3983	&0.4001	&0.1116 &0.0964	   &\textbf{0.7560}\\
&\uppercase\expandafter{\romannumeral2} &0.9469 &0.5161	&0.1872	&0.6195	&0.6175	&0.3285 &0.3330	   &\textbf{0.8434}\\
&\uppercase\expandafter{\romannumeral3} &0.9324 &0.3999	&0.0675	&0.2326	&0.2430	&0.1729 &0.1745	   &\textbf{0.7192}\\
&\uppercase\expandafter{\romannumeral4} &0.9389 &0.3382	&0.0386	&0.3714	&0.3497	&0.0803 &0.0725	   &\textbf{0.6850}\\
\hline
\multirow{4}*{Image  Cara}
&\uppercase\expandafter{\romannumeral1}	&0.7624 &0.5932 &0.1991 &0.8394	&0.8499 &0.8041 &0.7584	  &\textbf{0.8500}\\
&\uppercase\expandafter{\romannumeral2}	&0.8024 &0.7220 &0.6606 &\textbf{0.8923}	&0.8907 &0.8069 &0.7953	  &0.8916\\
&\uppercase\expandafter{\romannumeral3}	&0.4524 &0.3382 &0.1682 &0.0828	&0.7135 &\textbf{0.7463} &0.6085	  &0.7302\\
&\uppercase\expandafter{\romannumeral4}	&0.7625 &0.6328 &0.2010 &0.1433	&0.8457 &0.7993 &0.7511	  &\textbf{0.7998}\\
\hline
\multirow{4}*{Image  Frog}
&\uppercase\expandafter{\romannumeral1}	&0.6200 &0.5090 &0.5309 &0.8025	&0.8185 &0.6271 &0.5614	  &\textbf{0.8220}\\
&\uppercase\expandafter{\romannumeral2}	&0.6490 &0.6194 &0.7770 &0.8396	&0.8510 &0.6384 &0.5960	  &\textbf{0.8548}\\
&\uppercase\expandafter{\romannumeral3}	&0.4830 &0.3636 &0.3343 &0.2443	&0.7139 &0.5718 &0.4464	  &\textbf{0.7231}\\
&\uppercase\expandafter{\romannumeral4}	&0.6128 &0.5276 &0.5572 &0.4034	&0.7939 &0.6116 &0.5416	  &\textbf{0.8003}\\
\hline
\end{tabular}
\end{table*}
\subsection{Visual comparisons}
\label{sec:Exp-Visual}
In terms of visual analysis, a color-based method IDZ and seven widely used and noteworthy methods QDPC, QDPA, Canny, Sobel, Prewitt, Differential Phase Congruence (DPC) and Modified Differential Phase Congruence (MDPC)  will be compared with our algorithm.

\subsubsection{Color-based algorithm}
\label{sec:color}
In this part, we first compare the proposed algorithm with the IDZ gradient algorithm. In order to make the experiment more convincing, we used Gaussian filter before IDZ algorithm to achieve the effect of denoising.
Fig. \ref{fig8} presents the edge map of the noiseless House image, while Fig. \ref{fig9} presents the edge map of the House image corrupted with four different types of noise.
It can be seen from the second row of Fig. \ref{fig9} that IDZ gradient algorithm performs well in the first two images of the first line, while poorly in the last two images.
This illustrates that the IDZ gradient algorithm's limitations as a edge detector. The third row of Fig. \ref{fig9} shows the detection result of the proposed algorithm. It preserves details more clearly than the second row. It demonstrates that the proposed algorithm gives robust performance compared to that of the IDZ gradient algorithm.

\subsubsection{Grayscale-based algorithms}
\label{sec:Grayscale}
\begin{table*}[th]
\centering
\caption{The PSNR comparison values for Fig. \ref{test11}.
 Types of noise: I- Gaussian noise, II- Poisson noise, III- Salt \& Pepper noise and IV- Speckle noise.}
\label{tabP-11}
\begin{tabular}{cccccccccc}
\hline
&  &QDPC &QDPA   &Canny &Sobel &Prewitt   &DPC &MDPC   &Ours\\
\hline
&\uppercase\expandafter{\romannumeral1} &{59.5875} &{57.1603}	&{54.8202}	&{61.4068}	&{61.8794}  &{58.5991} &{59.1803}	   &\textbf{62.5796}\\
&\uppercase\expandafter{\romannumeral2} &{59.7487} &{58.0289}	&{56.0615}  &{61.7013}	&{62.0238}  &{59.1115} &{59.5553}	   &\textbf{64.3704}\\
&\uppercase\expandafter{\romannumeral3} &{58.9160} &{56.4631}	&{54.2282}	&{59.9397}	&{60.7958}	&{57.5608} &{58.4340}	   &\textbf{61.6690}\\
&\uppercase\expandafter{\romannumeral4} &{59.4685} &{57.4581}	&{54.7294}	&{60.3342}	&{61.5622}	&{58.4213} &{59.1584}	   &\textbf{62.9552}\\
\hline
\end{tabular}
\end{table*}
\begin{table*}[th]
\centering
\caption{The SSIM comparison values for Fig. \ref{test11}.
 Types of noise: I- Gaussian noise, II- Poisson noise, III- Salt \& Pepper noise and IV- Speckle noise.}
\label{tabS-11}
\begin{tabular}{cccccccccc}
\hline
&  &QDPC &QDPA &Canny &Sobel &Prewitt  &DPC &MDPC   &Ours\\
\hline
&\uppercase\expandafter{\romannumeral1} &{0.5250} &{0.3732}	&{0.1650}	&{0.5817}	&{0.6089}   &{0.4226} &{0.4711}	   &\textbf{0.7276}\\
&\uppercase\expandafter{\romannumeral2} &{0.5292} &{0.4451}	&{0.3232}   &{0.6508}	&{0.6642}   &{0.4768} &{0.5207}	   &\textbf{0.8188}\\
&\uppercase\expandafter{\romannumeral3} &{0.3871} &{0.2828}	&{0.0920}	&{0.3426}	&{0.3897}	&{0.2471} &{0.3543}	   &\textbf{0.6602}\\
&\uppercase\expandafter{\romannumeral4} &{0.5160} &{0.3899}	&{0.1673}	&{0.4769}	&{0.5877}	&{0.4103} &{0.4779}	   &\textbf{0.7536}\\
\hline
\end{tabular}
\end{table*}
We compare the performance of the proposed algorithm  with seven widely used and noteworthy algorithms. The  noiseless (Fig. \ref{fig5}) and noisy images (Fig. \ref{fig6}) are both taken into consideration. Here, the commonly used color-to-gray conversion formula \cite{gray1, gray2} is applied in the experiments, which is defined as follows
\begin{equation}
  Gray=0.299*R+0.587*G+0.114*B.
\end{equation}

\begin{itemize}
\item {\bf Noiseless case: }
  Fig. \ref{fig5} shows the eight noiseless test images.  Fig. \ref{6noiseless} demonstrates the edge detection results of the noiseless test images of Lena, Men, House and T1. Different rows correspond to the results of different methods. From top to bottom they are QDPC, QDPA, Canny, Sobel, Prewitt, DPC, MDPC and the proposed algorithms, respectively. While Fig. \ref{T1-3} demonstrates the edge maps of the noiseless test images T2, T3, Cara and Frog.
\item {\bf Noisy case: }
  Fig. \ref{fig6} is produced by adding four noises (I-IV) to each image in Fig. \ref{fig5}. The edge maps obtained by applying the QDPC, QDPA, Canny, Sobel, Prewitt, DPC, MDPC and the proposed methods to noisy images T1, T2 and Frog (Fig. \ref{fig6}) are shown in Fig. \ref{nT1}, Fig. \ref{nT2} and Fig. \ref{nFrog}, respectively.
\label{noise}
\end{itemize}

The bottom row of Fig. \ref{nT1}, Fig. \ref{nT2} and Fig. \ref{nFrog} respectively shows the edge results of the noisy image T1, T2 and Frog (Fig. \ref{fig6}) using our proposed method. we can clearly see that the proposed algorithm is able to extract edge maps from  the noisy images. This means that the proposed algorithm is resistant to the noise.  In particular, it is superior to the other detectors on images with Salt \& Pepper noise.

\subsection{Quantitative analysis}

The {PSNR} \cite{psnr} is a widely used method of objective evaluation of two images. It is based on the error-sensitive image quality evaluation.
In addition, the SSIM \cite{ssim} is a method of comparing two images under the three aspects of brightness, contrast and structure.
To show the accuracy of the proposed edge detector, the PSNR  and SSIM  values of various type of edge detectors on noisy images (I- Gaussian noise, II- Poisson noise, III- Salt \& Pepper noise and IV- Speckle noise) are calculated (Table \ref{tabP1} - \ref{tabS2}).\\
Tables \ref{tabP1} - \ref{tabS2} give the comparison results of the PSNR and SSIM values of the test images. Each value in the table represents the similarity between the edge map of the noisy image and the edge map of the noiseless image. That is, the larger the value, the stronger the denoising ability. From the results in Tables \ref{tabP1} and \ref{tabP2}, we obtain the following conclusions.

\begin{itemize}
\item Image Lena, Men, House and T1 results in Table \ref{tabP1} show that the top three algorithms are clearly, that is Sobel, Prewitt and the the proposed one. This shows that these three algorithms can achieve high similarity between the edge map of noisy image and noiseless image. Therefore, from the point of view of PSNR value, these three algorithms have excellent robustness than the others.

\item In Table \ref{tabP2}, for image T2, the top three algorithms are QDPC, QDPA and the proposed algorithm. For images Cara and Frog, although Sobel and DPC also performed well, it is not difficult to see that Prewitt and the proposed algorithm are more robust to the four type of noises than the others. While for image T3, the proposed method has also shown satisfactory performance for Gaussian noise and Speckle noise. On the whole, using the proposed method to do color edge detection on this type of image, the performance is obviously excellent.
\end{itemize}
Tables  \ref{tabS1} and \ref{tabS2} show the SSIM values between the edge maps of noiseless images and the edge maps of the noisy images. The closer the SSIM value is to 1, the better performance of the algorithm is. From the SSIM  values in these tables, we obtain the following conclusions.
\begin{itemize}
\item  From the SSIM values in Table \ref{tabS1}, our proposed algorithm gives  better performance than the other algorithms. For Poisson noise, the noise reduction effect of Sobel, Prewitt and the proposed algorithm are more robust than the other five methods. For Salt \& Pepper noise, the proposed method has the best noise reduction performance.
    While for Gaussian noise and Speckle noise, the proposed menthod is still in the top three. Therefore, from the perspective of SSIM value, the denoising performance of the proposed method is maintained very well.

\item Table \ref{tabS2} shows that, for image T2 and T3, the top three algorithms are QDPC and QDPA and the proposed algorithm. For image Cara and Frog, the top three algorithms are Sobel, Prewitt and the proposed algorithm. In general, the proposed algorithm SSIM values are always in the top three. Therefore, the noise immunity of the proposed method is optimal.
\end{itemize}
Tables  \ref{tabP-11} and \ref{tabS-11} show the PSNR and SSIM values between the edge maps of noiseless images and the edge maps of the noisy images, respectively. Each value in these tables is an average of the results for all the images in Fig. \ref{test11}. In particular, the images in Fig. \ref{test11} are randomly selected from the public dataset (http://decsai.ugr.es/cvg/dbimagenes/). From the result values in Tables  \ref{tabP-11} and \ref{tabS-11}, we find that the proposed method can achieve superior performance over state-of-the-art methods to detect edge, which demonstrates the effectiveness and feasibility of their practical use. It is more robust against noises.

\section{Conclusions and Discussions}
\label{sec:Con}

In this paper, we have proposed QHF as an effective tool for color image processing. Different from quaternion analytic signal, the QHF contains two parameters  that offers flexible perspective to deal with different color noisy images. Based on QHF and the improved Di Zenzo gradient operator, we proposed a new edge detection algorithm. Several experiments including visual comparison and quantitative analysis are conducted in the paper to verify the effectiveness of the proposed color image edge detection algorithm. However, the noisy images considered in this article each only involved one single kind of noise disturbance. In the future, further speed optimization needs to be invested and mixed types of noise \cite{B44}-\cite{3} situation should be considered for test.

%
\appendices


\ifCLASSOPTIONcaptionsoff
  \newpage
\fi



%

\end{document}